\documentclass[runningheads,preprint]{elsarticle}

\usepackage{amssymb}
\usepackage{amsmath}

\usepackage[T1]{fontenc}

\usepackage{graphicx}
\usepackage{tikz}
\usetikzlibrary{shapes.geometric}
\usepackage{todonotes}
\usepackage{soul}
\usepackage{hyperref}
\usepackage{algorithmicx}
\usepackage{algorithm, algpseudocode}
\usepackage{xspace}
\usepackage{enumitem}

\newcommand{\pp}{\textsc{Path Partition}\xspace}
\newcommand{\pc}{\textsc{Path Cover}\xspace}

\usepackage[most]{tcolorbox}
\newcommand{\problemdef}[4]{
	\begin{tcolorbox}[width = \textwidth,colback=white,arc=0pt,outer arc=0pt,boxrule=0.7pt,left =0.5em,right=0em]#1 #2		\\[2pt]
		\begin{tabular}{ @{}l p{0.84\textwidth} c }
			\textbf{Input:} & #3 \\[.5pt]
			\textbf{Task:} & #4
		\end{tabular}
	\vspace{-0.25em}
	\end{tcolorbox}
 }

\usepackage{amsthm}

\newtheorem{theorem}{Theorem}
\newtheorem{lemma}[theorem]{Lemma}
\newtheorem{corollary}[theorem]{Corollary}
\newtheorem{proposition}[theorem]{Proposition}
\newdefinition{definition}[theorem]{Definition}
\begin{document}

\begin{frontmatter}



\title{Polynomial-time algorithms for \pc and \pp on trees and graphs of bounded treewidth\tnoteref{conf}} 
\tnotetext[conf]{A preliminary version of the paper appeared in the proceedings of the CALDAM~2025 conference \cite{DBLP:conf/caldam/FoucaudMMR25}.}

\author[1]{Florent Foucaud,\fnref{fn1}}
\fntext[fn1]{Funded by French government IDEX-ISITE initiative 16-IDEX-0001 (CAP 20-25), International Research Center "Innovation Transportation and Production Systems" of the I-SITE CAP 20-25, and ANR project GRALMECO (ANR-21-CE48-0004).}

\author[2]{Atrayee Majumder,\fnref{fn2}}
\fntext[fn2]{Supported by Academic Career in Pilsen Program, 2025 cycle.}

\author[3]{Tobias M{\"{o}}mke,\fnref{fn3}}
\fntext[fn3]{Partially supported by DFG Grant 439522729 (Heisenberg-Grant)}

\author[3]{Aida Roshany-Tabrizi}

\affiliation[1]{organization={CNRS, Clermont Auvergne INP, Mines Saint-{\'{E}}tienne, LIMOS}, postcode={63000}, city={Clermont-Ferrand}, country={France}.}

\affiliation[2]{organization={Dept. of Mathematics, University of West Bohemia in Pilsen}, country={Czech Republic}.}

\affiliation[3]{organization={University of Augsburg}, country={Germany}.}
\begin{abstract}
    In the \pc problem, one asks to cover the vertices of a graph using the smallest possible number of (not necessarily disjoint) paths. While the variant where the paths need to be pairwise vertex-disjoint, which we call \pp, is extensively studied, surprisingly little is known about \pc. We start filling this gap by designing a linear-time algorithm for \pc on trees. 
    We show that \pc can be solved in polynomial time on graphs of bounded treewidth using a dynamic programming scheme. It runs in XP time $n^{t^{O(t)}}$ (where $n$ is the number of vertices and $t$ the treewidth of the input graph) or $\kappa^{t^{O(t)}}n$ if there is an upper-bound $\kappa$ on the solution size. A similar algorithm gives an FPT $2^{O(t\log t)}n$ algorithm for \pp, which can be improved to (randomized) $2^{O(t)}n$ using the Cut\&Count technique. These results also apply to the variants where the paths are required to be induced (i.e. chordless) and/or edge-disjoint.
\end{abstract}



\end{frontmatter}

\section{Introduction}

Path problems in graphs are fundamental problems in algorithmic graph theory, consider for example the problem of computing shortest paths in a graph, which has been one of the first studied graph problems for which efficient algorithms were obtained~\cite[Chapter 24]{CLRS3}. Finding \emph{disjoint paths} is also a problem of utmost importance, both in algorithmic and structural graph theory~\cite{DBLP:journals/jct/RobertsonS95b}. When it comes to \emph{covering} the graph, the \textsc{Hamiltonian Path} problem is another classic path-type problem studied both in combinatorics and computer science. We study its generalizations, \pc and \pp, which are about covering the vertices of a path using a minimum number of paths (unrestricted and pairwise vertex-disjoint, respectively). They are formally defined as follows.

\problemdef{\pc}{}{A graph $G$.}{Compute a minimum-size \emph{path cover}, that is, a set of paths of $G$ such that every vertex of $G$ belongs to at least one of the paths.}

Its variant that requires the solution paths to be pairwise vertex-disjoint is very well-studied and defined as follows.

\problemdef{\pp}{}{A graph $G$.}{Compute a minimum-size \emph{path partition}, that is, a set of pairwise vertex-disjoint paths of $G$ such that every vertex of $G$ belongs to exactly one of the paths.}

Our goal is to study the algorithmic complexity of \pc and \pp on trees and graphs of bounded treewidth. Treewidth is an important graph parameter and the associated tree decompositions enable to solve various problems efficiently when this parameter is bounded. 
We refer to the book~\cite{Paralg} for more on the topic of algorithms for graphs of bounded treewidth.

\paragraph{Related work} As both problems generalize \textsc{Hamiltonian Path} (which amounts to decide whether a graph can be covered by a single path), they are both NP-hard, and this holds even, for example, for 2-connected cubic bipartite planar graphs~\cite{MR596313}.

Most work in the literature refers to \pp as ``\pc'', which is unfortunately misleading. \pp is much more studied than \pc, see~\cite{PC-cocomp,DBLP:journals/dam/HungC07} for examples of such works on \pp. In some cases, \pp was also called \textsc{Hamiltonian Completion}~\cite{Franzblau_Raychaudhuri_2002,GH74,KUNDU197655}. To avoid any confusion between the two problems, we use the terminology of \pc and \pp as defined above, a choice taken from the survey~\cite{manuel2018revisiting} on path-type problems. 
Also refer to the PhD thesis~\cite{Baker13}.

Several papers from the 1970s studied \pp on trees~\cite{boesch1974covering,GH74,KUNDU197655}. 
However, they did not explicitly analyze the running times of their algorithms. 
A properly analyzed linear-time algorithm was given in 2002~\cite{Franzblau_Raychaudhuri_2002}. \pp was also shown to be solvable in polynomial time on many other graph classes such as cographs~\cite{LIN199575}, distance-hereditary graphs~\cite{DBLP:journals/dam/HungC07}, co-comparability graphs~\cite{PC-cocomp} (which contain all interval graphs), 
or block graphs/cactii~\cite{DBLP:journals/dam/PanC05}. We do not know of any explicit algorithm for \pp on graphs of bounded treewidth, but algorithms exist for the closely related \textsc{Cycle Partition} problem~\cite{DBLP:journals/talg/CyganNPPRW22}.


Both \pc and \pp have numerous applications, in particular in program testing~\cite{NH79}, circuit testing~\cite{PC95}, or machine translation~\cite{LIN200673}, to name a few. Although \pp is more widely studied, \pc is also a natural problem, with specific applications in bio-informatics when restricted to directed acyclic graphs~\cite{dagPC,DBLP:journals/bmcbi/RizziTM14}.

We refer to~\cite{dagPC,9628018,NH79,DBLP:journals/bmcbi/RizziTM14} for the few references about \pc that we are aware of. 

\paragraph{Our results} Although \pp is well-studied on trees and other graph classes, surprisingly, this is not the case of \pc. Note that despite the two problems having similar statements, they typically have very different optimal solutions. For example, on a star with $k$ leaves, an optimal path partition has size $k-1$, but an optimal path cover has size $\lceil k/2\rceil$.

We first focus (in Section~\ref{sec:trees}) on \pc on trees, for which no linear-time algorithm has been given in the literature. For trees, the size of an optimal solution is given by the ceiling of half of the number of leaves. The proof of this fact was given by Harary and Schwenk in 1972~\cite{article1} for the problem of covering the \emph{edges} of the tree. An analysis of their proof yields a quadratic-time algorithm. We show how \pc can in fact be solved in linear time on trees, by giving an improved algorithm based on depth-first-search (DFS).

We then study graphs of bounded treewidth in Section~\ref{sec:tw}. We design an explicit dynamic programming algorithm that runs in time $n^{t^{O(t)}}$ for graphs of treewidth at most $t$ and order $n$. In fact, this algorithm runs in time $\kappa^{t^{O(t)}}n$, where $\kappa$ is the maximum number of allowed paths in a solution (the "solution size" if we reformulate \pc as a decision problem). With a slight simplification, the same algorithm also solves \pp and runs in improved FPT running time $2^{O(t\log t)}n$. However, we show in Section~\ref{sec:PP} how to obtain a randomized algorithm for \pp running in time $2^{O(t)}n$ using the Cut\&Count technique.

It is not clear whether \pc can be solved in FPT time as well or not; however, we give some indications of why that might not be the case.

Moreover, we argue that our algorithms also apply to the versions of both \pc and \pp where the paths in the solution are required to be \emph{induced} (i.e. chordless) or pairwise edge-disjoint. These variants have been studied in the literature (see~\cite{fernau23,manuel2018revisiting} and references therein).

We finally conclude in Section~\ref{sec:conclu}.
\newpage
\section{A linear-time algorithm for \pc on trees}\label{sec:trees}

We first study \pc on trees. In~\cite[Theorem 7]{article1}, it is proved that the minimum number of paths needed to cover the \emph{edges} of a tree is equal to $\lceil \ell/2\rceil$, where $\ell$ is the number of leaves. This is an obvious lower bound, since for any leaf of a tree, there must be a solution path starting at that leaf. This also holds for covering the \emph{vertices}. The argument of~\cite{article1}, based on pairing the leaves arbitrarily and switching the pairing to increase the number of covered vertices at each step, leads to a quadratic-time algorithm for \pc on trees. We next present an algorithm for solving \pc of a tree with a runtime that is linear in the number of vertices.

\begin{theorem}
\label{thm:PCTL}
\pc can be solved in linear time on trees, and the optimal size of a solution for a tree with $\ell$ leaves is $\lceil\ell/2\rceil$.    
\end{theorem}
Consider the input tree $T$ to be rooted at an arbitrary internal vertex $r$ of $T$.
The intuition here is to cover the vertices of the tree by simulating the Depth-First-Search (DFS) algorithm with some modifications in the steps, thus the running time would become the same as the running time of DFS.\\
First, we recall the recursive DFS algorithm, see e.g.~\cite[Chapter 22.3]{CLRS3}. 
In this algorithm, there are two timestamps assigned to each vertex: the first timestamp is given when we discover the vertex for the first time while traversing the tree, and the second timestamp is given when we finish traversing all the vertices of the sub-tree rooted at that particular vertex. We use these timestamps in our algorithm.
We use the first timestamp to mark the vertex $v$ as visited and the second timestamp to consider a valid solution of \pc to cover the vertices rooted at $v$, including $v$.
The algorithm starts at the root and marks the vertices as visited in DFS order until it reaches a leaf. 
At the leaf on the way back, it starts a path with endpoints at the leaf and continues toward the parent, and the paths are recorded.
Now we need some definitions to specify the steps of our algorithm.

We use $p=(x_1,x_2,\dots,x_k)$ as a notation for paths where $x_i$ are all distinct vertices of $T$ connected by the path $p$. 
The vertices $x_2,\dots,x_{k-1}$ are internal vertices; the vertices $x_1$ and $x_k$ are endpoints of the path $p$.

\begin{definition}
 A path is \emph{closed} if both of its endpoints are leaves, and a path is \emph{open} if one of its endpoints is not a leaf. 
 We assume a singleton path at a leaf of the tree to be an open path.
 \end{definition}

\begin{definition}\label{Co-Ex}
     Let $v$ be a vertex in the tree $T$ and $\mathcal{P}$ be a set of paths in the subtree rooted at $v$.
     For instance $p_1=(a_1,\dots,a_j)$ and $p_2=(b_1,\dots,b_k)$ are two paths in $\mathcal{P}=\{p_1,p_2,\dots\}$ 
     where $a_1$ and $b_1$ are leaves of $T$.
    We define the following operations:
    \begin{itemize}[nolistsep]
        \item We use the notation $\mathcal{P}\cdot v$ to \emph{concatenate} the paths in $\mathcal{P}$ with vertex $v$, that is, For each path $p\in \mathcal{P}$ if there exists an edge between $v$ and the endpoint of $p \in \mathcal{P}$, then we add vertex $v$ to the path $p$.
        As an example $\mathcal{P}\cdot v=\{(a_1,\dots,a_j,v),(b_1,\dots,b_k,v),\dots\}$.
        If $\mathcal{P}=\{p\}$ for a single path $p$, we write $p\cdot v$ as shorthand for $\mathcal{P}\cdot v$.
        \item Two vertex-disjoint open paths $p_1=(a_1,\dots,a_j)$ and $p_2=(b_1,\dots,b_k)$ are \emph{combined} at the vertex $v$ if $v$ is connected to one of the endpoints of each path.
        By combining $p_1$ and $p_2$ at vertex $v$, we obtain a new path, formally, $comb(p_1,p_2)=(a_1,\dots,a_j,v,b_k,\dots,b_1)$.   
    \end{itemize}
\end{definition}

Let $T_v$ be the subtree of $T$ rooted at a given vertex $v$. Let $\{T_v^1, T_v^2, \ldots\}$ be the set of connected components of $T_v \setminus \{v\}$. Two open paths coming from these connected components to $v$ can only be combined at $v$ if they are coming from two different components from $\{T_v^1, T_v^2, \ldots\}$. Since, for the combining operation, the paths need to be vertex-disjoint, the open paths coming from the same component to $v$ can not be combined at $v$. Two open paths are called \emph{unrelated} if they come from two different components from $\{T_v^1, T_v^2, \ldots\}$ at $v$.

In the algorithm, for each vertex $v \in T$, we define three sets of paths: 
\begin{itemize}[nolistsep]
    \item $\mathcal{P}^{\emph{close}}_v$ is a set of paths that consists of \emph{closed} paths combined at vertex $v$,
    \item $\mathcal{P}^{\emph{open}}_v$ is a set of paths that consists of \emph{open} paths that will be extended further from vertex $v$,
    \item $\mathcal{P}_v$ is the set of all paths that are extended as open paths from the children of $v$.
\end{itemize} 
Among these sets, we need to store the set $\mathcal{P}^{\emph{open}}_v$ corresponding to each node for further usage in the parent node of $v$.

The aim of our algorithm is to have all paths closed, and when a path is open it means we extend the path until it gets combined and closed. 
Note that we mark a vertex covered when all the vertices of the subtree rooted at that vertex are visited and covered. Now, we present Algorithm~\ref{Alg-PCTL} to compute a path cover of $T$. The input to our algorithm is a tree $T$, a designated root vertex $r$, and a solution set $S$, which is initially $\emptyset$. After the execution of Algorithm~\ref{Alg-PCTL}, the solution set $S$ provides an optimal path cover of $T$. We store a path $p = (a_1, a_2, \ldots, a_j)$ in terms of its endpoints i.e. $p = (a_1,a_j)$.

\begin{algorithm}[H]
\caption{$PC(T,v,S)$\label{Alg-PCTL}}
\begin{algorithmic}[1]
\State Mark $v$ as visited
\For{Each child $u$ of $v$}
\If{$u$ is unvisited} 
\State recursively call $PC(T,u,S)$
\EndIf
\EndFor
\If{$v$ is a leaf}
\State $\mathcal{P}^{\emph{open}}_v \longleftarrow \{(v,v)\}$ \hspace{0.5cm} and \hspace{0.5cm} $\mathcal{P}^{\emph{close}}_v \longleftarrow \emptyset$
\Else
\State $\mathcal{P}_v\longleftarrow \bigcup \limits_{c} \mathcal{P}^{\emph{open}}_c$ \Comment{taking a union over all the paths coming from all children $c$ of $v$}
\State $\mathcal{P}^{\emph{close}}_v \longleftarrow \emptyset$
\While{$|\mathcal{P}_{v}| > 2$ }
\State Find two unrelated paths $p_i$ and $p_j$
\State $\mathcal{P}_v^{\emph{close}}\longleftarrow \mathcal{P}_v^{\emph{close}} \bigcup \emph{comb}(p_i,p_j)$
\State $\mathcal{P}_v \longleftarrow \mathcal{P}_v \setminus \{p_i,p_j\}$
\EndWhile
\If{$v$ is the root of $T$ and $|\mathcal{P}_v| = 2$}
\State $\mathcal{P}_v^{\emph{close}}\longleftarrow \mathcal{P}_v^{\emph{close}} \bigcup \emph{comb}(p_1,p_2)$
\ElsIf{$v$ is the root of $T$ and $|\mathcal{P}_v| = 1$}
\State $\mathcal{P}_v^{\emph{close}} \longleftarrow \mathcal{P}_v^{\emph{close}} \bigcup p_1 \cdot v$
\Else
\State $\mathcal{P}_v^{\emph{open}} \longleftarrow \mathcal{P}_v \cdot v$ \Comment{when $v$ is not the root of $T$ and $|\mathcal{P}_v| \leq 2$}
\EndIf
\EndIf
\State Mark $v$ as covered
\State $S \longleftarrow S \bigcup \mathcal{P}_v^{\emph{close}}$
\State \Return $\mathcal{P}^{\emph{open}}_v$
\end{algorithmic}
\end{algorithm}
The algorithm begins its search from the root, marks each vertex \emph{visited}, and recursively calls itself until it reaches a leaf (say $v$) [line 1-6]. 
The algorithm starts a singleton path $(v,v)$ from the leaf and extends the path to its parent as an open path (line 8) and marks 
$v$ as \emph{covered}. 
  
During the backtracking process, at each visited vertex $v$, the algorithm marks $v$ as \emph{covered} when all its children are visited and \emph{covered}. 
For each vertex $v$, other than a leaf and the root of $T$, $\mathcal{P}_v$ contains all the open paths coming from its children. It is clear from line 22 that each child can forward at most two open paths to its parent. Therefore if $|\mathcal{P}_v| \geq 3$, we will always get two unrelated paths coming from two different children of $v$ that can be combined to create a closed path. In the while loop in (12-16) we precisely do it.
The algorithm constructs $\mathcal{P}^{\emph{close}}_v$ by combining all possible pairs of unrelated paths at $v$ and adding them to the final solution set $S$ (line 26).
 
The steps the algorithm follows for the root $r$ of $T$ are the same as any other internal vertex, except all the open paths are closed at this point. 
If the number of open paths in $\mathcal{P}_r$ is even, all the paths are combined at $r$ and closed (line 18). Note that $r$ is an internal vertex of $T$. Therefore, the number of children of $r \geq 2$. Hence, it is always possible to close all pairwise unrelated paths at $r$. 
If the number of paths in $\mathcal{P}_r$ is odd, we make $r$ as the endpoint of the only remaining open path, which cannot be combined with any other path (line 20). We are now ready to prove Theorem~\ref{thm:PCTL}.

\begin{proof}[Proof of Theorem~\ref{thm:PCTL}]
    First, we have to prove that the size of $S$ is $\big \lceil \frac{\ell}{2} \big \rceil$ where the number of leaves of $T$ is $\ell$. We start a path in the leaf of $T$ and keep it as an open path as long as it is not combined with another path and closed. When a path is closed, both of its endpoints become leaves of $T$.  
    Additionally, in the root, all the paths get closed except for possibly one path. Hence, for an even number of leaves, the size of $S$ is $\frac{\ell}{2}$ and for an odd number of leaves, the size of $S$ is $\frac{\ell - 1}{2}+1$, which proves the first part of the theorem.

    Now, we show that the running time of Algorithm~\ref{Alg-PCTL} is $O(n)$ where $n$ is the number of vertices of $T$. Algorithm~\ref{Alg-PCTL} closely resembles the DFS algorithm, with the addition of some constant time operations to store the paths covering the vertices.
    Two operations, $\emph{comb}(p,q)$ and $\mathcal{P}_v \cdot v$ take $O(1)$ time as we store the paths by their endpoints, and after each of these operations, only the endpoints are changed while the new paths are created. 
    Finding a pair of unrelated paths can also be done in $O(1)$ time, as the children of a particular vertex can be ordered from left to right, and the paths coming from the children can also be ordered accordingly. For each path $p$, we just need to check at most two consecutive paths in this ordering to find a path $q$ which is unrelated to $p$.
    Therefore, in each of the iterations of the while loop (line 12-16) we do the operations in $O(1)$ time. All the iterations of the while loop together traverse each of the edges between $v$ and its children $O(1)$ time. These edges are only traversed in $v$'s recursive call. 
  The recursive $PC(T,v,S)$ call is made only once for each vertex. Therefore, the recursive calls altogether make the running time of the algorithm $O(n+m)$, where $m$ is the number of edges of $T$ (since each edge of $T$ is traversed $O(1)$ times), which is same as the running time of DFS.
  As $m$ is $O(n)$ for $T$, the total running time becomes $O(n)$, which proves the second part of the theorem.
\end{proof}

\section{\pc on graphs of bounded treewidth}\label{sec:tw}

In this section, we present an algorithm that solves \pc on general graphs in XP time when parameterized by treewidth. The algorithm is a classic dynamic programming scheme over a tree decomposition.

\begin{theorem}\label{theoxp}
    \pc can be solved in time $\kappa^{t^{O(t)}}\cdot n $ on graphs of treewidth $t$ and order $n$, where $\kappa$ is  the maximum allowed solution size. 
\end{theorem}

Since we can always assume that $\kappa\leq n$, we have the following corollary.

\begin{corollary}\label{corxp}
\pc can be solved in time $n^{t^{O(t)}}$, where $n$ is the number of vertices, and $t$ is the treewidth of the input graph.    
\end{corollary}

\subsection{Preliminaries}

Here, we give the definitions of \emph{tree decomposition}, \emph{treewidth}, and a well-structured form of it, a \emph{nice tree decomposition} from~\cite{kloks1994treewidth}.

\begin{definition}
 Let $G$ be an undirected graph. A tree decomposition of $G=(V,E)$ is a pair $(\mathcal{T},X)$, where $\mathcal{T}$ is a tree and $X:V(\mathcal{T})\to2^{V}$ satisfies the following conditions:
\begin{itemize}[nolistsep]
    \item for each $e\in E$ there exists a node $v\in  V(\mathcal{T})$ such that both endpoints of $e$ are inside $X_v$, a bag associated with $v$.
    \item for each $u\in V$ the set $\{v\in V(\mathcal{T}):u\in X_v\}$ is connected in $\mathcal{T}$.
\end{itemize}
The \emph{width} of $(\mathcal{T},X)$ is $max_{v\in V(\mathcal{T})}|X_v|-1$. 
The \emph{treewidth} of a graph $G$, denoted by $t$, is the minimum width of all possible tree decompositions of $G$. 
\end{definition} We distinguish between vertices of $G$ and vertices of the tree decomposition by referring to the latter as \emph{nodes}. Each node $v$ is associated with a \emph{bag} $X_v$: a subset of vertices of graph $G$. We use the classic \emph{nice tree decompositions}:

\begin{definition}[\cite{kloks1994treewidth}]\label{nicetree}
    A tree decomposition $\mathcal{T}$ of graph $G$ is \emph{nice} if:
    \begin{itemize}[nolistsep]
\item $\mathcal{T}$ is rooted at node $r$ with $X_r=\emptyset$.
    \item Each leaf node $v$ of $\mathcal{T}$ has an empty bag $X_{v}=\emptyset$.
        \item Each non-leaf node is of one of the types below:
        \begin{itemize}
            \item[1.]An \textbf{Introduce node} $v$ has exactly one child $u$ such that $X_v=X_u\dot{\cup} \{x\}$ for some vertex $x$ in $V$.
            \item[2.]A \textbf{Forget node} $v$ has exactly one child $u$ such that $X_v=X_u\setminus\{x\}$ for some vertex $x$ in $V$.
            \item[3.]A \textbf{Join node} $v$ has two children $u_1$, $u_2$ such that $X_v=X_{u_1}=X_{u_2}$.
        \end{itemize}
    \end{itemize}
\end{definition}
Given a graph $G$ on $n$ vertices and an integer $t$, there is an algorithm that either outputs a tree decomposition of $G$ of width at most $2t+1$, or determines that the treewidth of $G$ is larger than $t$, in time $2^{O(t)}n$~\cite{DBLP:conf/focs/Korhonen21}. Given a tree decomposition of $G$ of width $t$ and $O(n)$ nodes, a nice tree decomposition of width $t$ with at most $4n$ nodes can be constructed in time $O(t\cdot n)$~\cite{bodlaender2013fine}.

Let $G=(V,E)$ be a graph and $(\mathcal{T},X)$ be a nice tree decomposition of $G$ with width at most $t$. Let $\mathcal{T}_v$ be the subtree of $\mathcal{T}$ rooted at node $v$ and $V_v$ be the union of all bags in this subtree including $X_v$.
We define $G_v$ as the subgraph of $G$ induced by $V_v$.
We will use the following lemma from~\cite{Paralg}.

\begin{lemma}
\label{sepa}
    Let $(\mathcal{T},X)$ be a tree decomposition of a graph $G$ and let $uv$ be an edge of $\mathcal{T}$.
    The forest $\mathcal{T}-uv$ obtained from $\mathcal{T}$ by deleting edge $uv$ consists of two connected components $\mathcal{T}_u$ and $\mathcal{T}_v$.
    Then $(V_u,V_v)$ is a separation of $G$ with separator $X_u\cap X_v$.
\end{lemma}

\subsection{The dynamic programming scheme}

We implement our dynamic programming scheme for \pc in a bottom-up manner, starting at the leaf nodes of $\mathcal{T}$.
At each node $v$, we deal with a potential solution for $G$ that covers the vertices of $X_v$.
We define a path as a sequence of vertices where any vertex in a path $p$ is incident to at most two vertices from $p$, giving at most two neighbours in the subgraph induced by $p$. 
Each neighbour of a vertex in $p$ can either belong to $G_v-X_v$ (and thus have been forgotten in the bag of some descendant node of $v$), which we represent with ``$\downarrow$'', or belong to $G-G_v$ (it will appear in the bag of some ancestor node of $v$), which we represent by ``$\uparrow$'', or belong to $X_v$, which we represent by ``$-$''.
Therefore, we characterize each vertex by the types of its neighbour(s) inside a potential solution path. 
To this end, let $\mathcal{N}$ be the set of multisets of size at most two whose elements are taken from the set $\{-,\uparrow,\downarrow\}$. 

At each node $v$, we define a \emph{partial path} $p$ (representing the intersection of a path of $G$ with the bag $X_v$) by $(\Pi_p,\varphi_p)$ where $\Pi_p$ is an ordered subset of $X_v$ and $\varphi_p$ is a function $\varphi_p:\Pi_p\to \mathcal{N}$ which describes, for each vertex of $p$, the types of its neighbour(s) inside the path of $G$ represented by $p$. 
These paths are called ``partial'' since they consist of several paths, and might have neighbour(s) in the path that lie outside the node. 
Each vertex $x$ of a partial path $p$ in $X_v$ has the possibility of being in either of the \emph{types} illustrated in Figure~\ref{fig:enter-label} with respect to the path $p'$ of $G$ represented by $p$. 
\begin{definition}
The set $\mathcal N$ of types of a vertex with respect to a partial path $p$ of a node $v$ are as follows:    
\begin{itemize}[nolistsep]
    \item[1.] $p=p'$ is a one-vertex path (and thus $x$ has no neighbours in $p'$); $x$ is assigned the type $\emptyset$.
    \item[2.] $x$ is an endpoint of $p'$ (with a unique neighbour in $p'$), which is inside the bag; it is assigned the type $\{-\}$.
    \item[3.] $x$ is an internal vertex of $p'$ with two neighbours in $p'$, both inside the bag; it is assigned the type $\{-,-\}$
    \item[4.] $x$ is an endpoint of $p'$ with a unique neighbour in $p'$, in an ancestor bag; it is assigned the type $\{\uparrow\}$.
    \item[5.] $x$ is an endpoint of $p'$ with a unique neighbour in $p'$, in a descendant bag; it is assigned the type $\{\downarrow\}$.
    \item[6.] $x$ is an internal vertex of $p'$ with two neighbours in $p'$, both in ancestor bags; it is assigned the type $\{\uparrow,\uparrow\}$.
    \item[7.] $x$ is an internal vertex of $p'$ with two neighbours in $p'$ in descendant bags; it is assigned the type $\{\downarrow,\downarrow\}$.
    \item[8.] $x$ is an internal vertex of $p'$ with a neighbour in $p'$ in a descendant bag and a neighbour in $p'$ in an ancestor bag; it is assigned the type $\{\downarrow,\uparrow\}$.
    \item[9.] $x$ is an internal vertex of $p'$ with a neighbour in $p'$ in an ancestor bag and a neighbour in $p'$ inside the bag itself; it is assigned the type $\{\uparrow,-\}$.
    \item[10.] $x$ is an internal vertex of $p'$ with a neighbour in $p'$ in a descendant bag and a neighbour in $p'$ inside the bag itself; it is assigned the type $\{\downarrow,-\}$.
\end{itemize}
\end{definition}

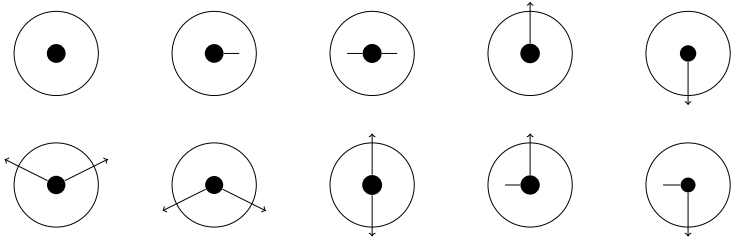
\begin{figure}[H]
\center
    \scalebox{0.7}{\begin{tikzpicture}
    \draw (0,0) circle (.8cm);
    \draw (3,0) circle (.8cm);
    \draw (6,0) circle (.8cm);
    \draw (9,0) circle (.8cm);
    \draw (0,-2.5) circle (.8cm);
    \draw (3,-2.5) circle (.8cm);
    \draw (6,-2.5) circle (.8cm);
    \draw (9,-2.5) circle (.8cm);
    \draw (12,-2.5) circle (.8cm);
    \draw (12,0) circle (.8cm);
      \node[circle, fill=black, inner sep=0pt] (A) at (0,0) {A};
      \node[circle, fill=black, inner sep=.1pt] (B) at (3,0) {B};
      \node[circle, fill=black, inner sep=.1pt] (C) at (6,0) {C};
      \node[circle, fill=black, inner sep=.1pt] (D) at (9,0) {D};
      \node[circle, fill=black, inner sep=.1pt] (E) at (0,-2.5) {E};
      \node[circle, fill=black, inner sep=.1pt] (F) at (3,-2.5) {F};
      \node[circle, fill=black, inner sep=.1pt] (G) at (6,-2.5) {G};
      \node[circle, fill=black, inner sep=.1pt] (H) at (9,-2.5) {H};
      \node[circle, fill=black, inner sep=.1pt] (I) at (12,-2.5) {I};
      \node[circle, fill=black, inner sep=.1pt] (J) at (12,0) {J};

      \node[circle, fill=white, inner sep=.1pt] (1) at (3.5,0) {};
      \node[circle, fill=white, inner sep=.1pt] (2) at (5.5,0) {};
      \node[circle, fill=white, inner sep=.1pt] (3) at (6.5,0) {};
      \node[circle, fill=white, inner sep=.1pt] (4) at (9,1) {};
      \node[circle, fill=white, inner sep=.1pt] (5) at (12,-1) {};
      \node[circle, fill=white, inner sep=.1pt] (6) at (-1,-2) {};
       \node[circle, fill=white, inner sep=.1pt] (7) at (1,-2) {};
      \node[circle, fill=white, inner sep=.1pt] (8) at (2,-3) {};
      \node[circle, fill=white, inner sep=.1pt] (9) at (4,-3) {};
      \node[circle, fill=white, inner sep=.1pt] (10) at (6,-1.5) {};
      \node[circle, fill=white, inner sep=.1pt] (11) at (6,-3.5) {};
      \node[circle, fill=white, inner sep=.1pt] (12) at (9,-1.5) {};
      \node[circle, fill=white, inner sep=.1pt] (13) at (8.5,-2.5) {};
      \node[circle, fill=white, inner sep=.1pt] (14) at (11.5,-2.5) {};
      \node[circle, fill=white, inner sep=.1pt] (15) at (12,-3.5) {};

       \draw (1) -- (B);

       \draw  (2) -- (C);
       \draw (C) -- (3);

       \draw[->] (D) -- (4);

       \draw[->] (J) -- (5);

       \draw[->] (E) -- (6);
       \draw[->] (E) -- (7);

       \draw[->] (F) -- (8);
       \draw[->] (F) -- (9);

       \draw[->] (G) -- (10);
       \draw[->] (G) -- (11);

       \draw[->] (H) -- (12);
       \draw (H) -- (13);

       \draw[->] (I) -- (15);
       \draw (I) -- (14);

\end{tikzpicture}}
    \label{fig:enter-label}
    \caption{Illustration of the different types of a vertex with respect to some node $v$ and partial path. The set of types is $\mathcal{N}=\{\emptyset, \{-\},\{-,-\},\{\uparrow\},\{\downarrow\},\{\uparrow,\uparrow\},\{\downarrow,\downarrow\},\{\downarrow,\uparrow\},\{\uparrow,-\},\{\downarrow,-\}\}$. ``$-$'' shows a vertex with a neighbour inside the bag $X_v$, ``$\uparrow$'' shows a vertex with a neighbour in $G-G_v$, and ``$\downarrow$'' shows a vertex with a neighbour in $G_v-X_v$.} \label{vertexposition}
\end{figure}

We say that a partial path for node $v$ is \emph{consistent} if it corresponds to the intersection with $X_v$ of some path of $G$.

Now, we extend the notion of a partial path $p$ of a node $v$ to that of a partial path $p$ of a subgraph $G_v$ in a natural way as follows. A partial path $p$ of $G_v$ is an ordered subset $\Pi_{p}$ of $V_v$ and a function $\varphi_p:\Pi_p\to \mathcal{N}$ which describes, for each vertex of a path, the types of its neighbour(s) inside the path, where ``$\uparrow$'' now refers to a neighbour in $V(G)\setminus V_v$, ``$-$'' refers to a neighbour in $V_v$, and there is no vertex of $p$ whose type contains ``$\downarrow$''.


We say that a partial path $p$ of $v$ \emph{agrees} with a partial path $p'$ of $G_v$ if the following conditions hold:
\begin{itemize}[nolistsep]
        \item $\Pi_p\subseteq \Pi_{p'}$;
        \item the order of the vertices in $\Pi_p\cap \Pi_{p'}$ is the same in $\Pi_p$ and $\Pi_{p'}$;
        \item all vertices in $\Pi_{p'}\setminus \Pi_{p}$, have a type in $\{-\},\{-,-\},\emptyset$;
        \item  $\varphi_p(x)=\varphi_{p'}(x)$, for each vertex $x$ with types $\emptyset$, $\{-\}$, $\{-,-\}$, $\{\uparrow\}$, and $\{\uparrow,-\}$;
        \item If $\varphi_p(x)=\{\downarrow\}$ then $\varphi_{p'}(x)=\{-\}$;
        \item If $\varphi_p(x)=\{\uparrow,\uparrow\}$ then $\varphi_{p'}(x)=\{\uparrow,\uparrow\}$;
        \item If $\varphi_p(x)=\{\downarrow,\downarrow\}$ then $\varphi_{p'}(x)=\{-,-\}$;
        \item If $\varphi_p(x)=\{\uparrow,\downarrow\}$ then $\varphi_{p'}(x)=\{\uparrow,-\}$;
        \item If $\varphi_p(x)=\{\downarrow,-\}$ then $\varphi_{p'}(x)=\{-,-\}$.
\end{itemize}

For a node $v$, we define a \emph{partial solution} $S$ of $G_v$ as a collection $P$ of partial paths of $G_v$ whose vertices cover all the vertices in $V_v$. We require that every partial path is consistent, that is, it may correspond to the intersection of a path of $G$ with $G_v$. We also require that any partial path that has no vertex whose type contains ``$\uparrow$'', is unique (i.e. appears only once in $S$). However, partial paths with vertices whose type contains ``$\uparrow$'' may appear multiple times (at most $n$ times). Indeed, they might correspond to different paths in a path cover of $G$, that all happen to have the same intersection with $G_v$. For notational convenience, a partial solution $S$ is stored as a set $P$ of partial paths, and a function $f$ where, for every partial path $p$ in $S$, we let $f(p)\in\{1,\ldots,n\}$ be the number of times the partial path $p$ appears in the partial solution $S$.

Let $P$ be a set of partial paths of a node $v$ together with a function $f:P\to\{1,\ldots,n\}$, and $S=(P',f')$ be a partial solution of $G_v$. Let $M$ be the multiset obtained from $P$ by adding, for every partial path $p$ of $P$, $f(p)$ copies of $p$ to $M$. Similarly, let $M'$ be the multiset obtained from $P'$ obtained from the subset of partial paths of $P'$ that intersect $X_v$, by adding to $M'$, $f'(p)$ copies of each such path $p$ of $P'$. We say that $P$ \emph{corresponds to $S$} if there is a bijection $\psi$ from the $M$ to $M'$, such that $p$ agrees with $\psi(p)$ for every partial path $p$ of $P$.


We are now ready to describe the dynamic programming scheme.

\paragraph{Dynamic Programming} We define a DP-state $[v,P,f]$ as follows:
\begin{itemize}[nolistsep]
    \item $v$ is a node in a nice tree decomposition $(\mathcal{T},X)$; 
    \item $P$ is a set of partial paths that covers the vertices of $X_v$;
    \item $f$ is a function $f:P\to\{1,\dots,n\}$ that maps each partial path $p$ in $P$ to a value that shows how many times $p$ is used.
\end{itemize}

We say that a DP-state $[v,P,f]$ \emph{valid} if every partial path in $P$ is consistent, and $(P,f)$ corresponds to a partial solution $S$ of $G_v$.

For a valid DP-state, we define $opt[v,P,f]$ as the minimum number of partial paths in a partial solution of $G_v$ corresponding to $(P,f)$. If there exists no such partial solution, then the DP-state is invalid, and $opt[v,P,f]=\infty$.

Now, we explain how to compute the value of a DP-state from the values of the children's DP-states.
We need to check the compatibility condition between DP-states of a node and the children's DP-sates; this is specific for each node type. Therefore, we consider each case individually. 

Let $v$ be an introduce node with a child node $u$.
Two valid DP-states $[v,P,f]$ and $[u,P',f']$ are \emph{compatible} if there exist partial solutions $S_v$ of $G_v$ and $S_u$ of $G_u$ corresponding to $(P,f)$ and $(P',f')$ respectively, such that the intersection of $S_v$ with $G_u$ is $S_u$.
Let $v$ be a forget node with a child $u$.
Two valid DP-states $[v,P,f]$ and $[u,P',f']$ are \emph{compatible} if there exist partial solutions $S_v$ of $G_v$ and $S_u$ of $G_u$ corresponding to $(P,f)$ and $(P',f')$ respectively, such that the intersection of $S_u$ with $G_v$ is $S_v$. 

Let $v$ be a join node with two children $u_1$ and $u_2$.
Three valid DP-states $[v,P,f]$, $[u_1,P',f']$, and $[u_2,P'',f'']$ are \emph{compatible} if there exist partial solutions $S_u,S_{u_1},S_{u_2}$ of $G_v$, $G_{u_1}$ and $G_{u_2}$ corresponding to $(P,f)$, $(P',f')$ and $(P'',f'')$ respectively, such that the intersection of $S_v$ with $G_{u_1}$ is $S_{u_1}$ and the intersection of $S_v$ with $G_{u_2}$ is $S_{u_2}$.

The algorithm proceeds as follows. It starts at the leaf nodes and approaches the root. At each node, the algorithm computes a value for each possible DP-state using the values of the children's compatible DP-states, and chooses the minimum possibility. 
We next describe the computations done at each node, depending on their type.

\paragraph{Leaf node} For each leaf node $v$, $X_{v}=\emptyset$, so it is trivial that $opt[v, \emptyset,f]=0$. 

\paragraph{Introduce node} Let $v$ be an introduce node with child node $u$ such that $X_v=X_u\dot\cup \{x\}$ for some $x\in V(G)$. Let $[v,P,f]$ be a valid DP-state for node $v$.

Note that if there is a partial path $p\in P$ and $x$ with one of the values $\varphi_p(x)=\{\downarrow\}$, $\varphi_p(x)=\{\uparrow,\downarrow\}$, $\varphi_p(x)=\{-,\downarrow\}$, and $\varphi_p(x)=\{\downarrow,\downarrow\}$, then the DP-state is not valid and $opt[v,P,f]=\infty$. Indeed, by Lemma~\ref{sepa}, there is no edge between $x$ and a vertex in $V_v\setminus X_v$. 

Otherwise, we will look for all DP-states $[u,P',f']$ for $u$ compatible with $[v,P,f]$. 
To check whether $[u,P',f']$ and $[v,P,f]$ are compatible, we only need to check the type of $x$ in every partial path $p$ (and the type of its neighbour(s) in $p$), since the remaining state in $X_u$ must be the same as $X_v$. 

Let $P_{\overline{x}}$ be the set of partial paths of $P$ that contains at least one vertex other than $x$. We must find a bijection between the partial paths of $P_{\overline{x}}$ and those of $P'$ (taking into account their multiplicities), in the following way. For every partial path $p$ of $P$ not containing $x$ at all, we require that $p$ also belongs to $P'$, with $f(p)=f'(p)$. Moreover, for every partial path $p$ of $P$ containing $x$ and at least one other vertex of $X_v$, there must be a partial path $p'$ of $P'$ such that $p'$ agrees with $p$ after removing $x$. 

Let $C$ be the collection of all DP-states $[u,P',f']$ that are compatible with DP-state $[v,P,f]$ and let $k$ be the number of partial paths of $P$ (accounting for their multiplicity via function $f$) containing only $x$. We have:
    \[
    opt[v,P,f]=\min_{\substack{[u,P',f']\in C}}\{opt[u,P',f']\}+k
    \]

To justify the validity of this formula, note that (by induction hypothesis) our process constructs only valid partial solutions. Indeed, consider a partial solution $S'$ corresponding to $(P',f')$ of size $opt[u,P',f']$, where $[u,P',f']\in C$ is such that $opt[u,P',f']$ is minimum. We obtain a partial solution $S$ of size $|S'|+k$ corresponding to $(P,f)$ by extending $S'$. To do so, include vertex $x$ into the partial paths of $P'$ that correspond to partial paths of $P$ containing $x$, as explained in the compatibility check. Moreover, add as many singleton partial paths $x$ as needed. This shows that $opt[v,P,f]\leq |S|\leq |S'|+k=\min_{\substack{[u,P',f']\in C}}\{opt[u,P',f']\}+k$.

Conversely, note that an optimal partial solution $S$ of size $opt[v,P,f]$ corresponding to $(P,f)$ can be transformed into another one, $S'$, by deleting $x$ from it, with $|S'|=|S|-k$ and $S'$ corresponds to some $(P',f')$ where $[u,P',f']$ is compatible with $[v,P,f]$. This shows that $\min_{\substack{[u,P',f']\in C}}\{opt[u,P',f']\}\leq |S'|\leq opt[v,P,f]-k$ and thus, $opt[v,P,f]\geq \min_{\substack{[u,P',f']\in C}}\{opt[u,P',f']\}+k$.

\paragraph{Forget node} Let $v$ be a forget node with child $u$ such that $X_v=X_u\setminus \{x\}$ for some $x\in V$.
Let $[v,P,f]$ be a valid DP-state of $X_v$.

Again, we will look for all DP-states $[u,P',f']$ compatible with $[v,P,f]$. To do so, similarly as in an introduce node, we need to check the type of $x$ (and its neighbours in the bag) for every partial path $p$ of $[u,P',f']$. We observe that in $[u,P',f']$, $x$ cannot have a type containing ``$\uparrow$'' since by Lemma~\ref{sepa}, $x$ cannot be adjacent to a vertex of $V(G)\setminus V_v$, so this type of DP-states can be disregarded. For every partial path $p$ in $P$, we need to find a partial path $p'$ of $P'$ such that $p'$ agrees with $p$, that is, either $p=p'$, or $p'$ is obtained from $p$ by removing $x$. Moreover, the multiplicities of the partial paths of $P$ corresponding to those of $P'$ have to match. 
Moreover, the only partial paths of $P'$ not corresponding to any partial paths of $P$ are those containing only $x$; every other partial path of $P'$ must correspond to one of $P$.  

Let $C$ be the collection of all DP-states $[u,P',f']$ that are compatible with $[v,P,f]$. We have:
    \[
    opt[v,P,f]=\min_{\substack{[u,P',f']\in C}}\{opt[u,P',f']\}
    \]

As for the introduce node, the upper bound follows because one can construct a partial solution of $G_v$ corresponding to $[v,P,f]$ using the partial solutions for $G_u$ in $C$. The lower bound follows because one can, conversely, create a partial solution for $G_u$ corresponding to some DP-state of $C$ by using one for $G_v$ corresponding to $[v,P,f]$.

\paragraph{Join node} Let $v$ be a join node with children $u_1$ and $u_2$ such that $X_v=X_{u_1}=X_{u_2}$. Let $[v,P,f]$ be a valid DP-state for $v$.

As before, we will check the compatibility of $[v,P,f]$ with (valid) DP-states $[u_1,P',f']$ and $[u_2,P'',f'']$. Similarly as before, the DP-states $[v,P,f]$, $[u_1,P',f']$, and $[u_2,P'',f'']$ are compatible if there is a bijection $\psi_1$ between the partial paths of $P$ and those of $P'$ (again, accounting for their multiplicities, that is, we consider $f(p)$ copies of each partial path $p$), and a similar bijection $\psi_2$ between the partial paths of $P$ and those of $P''$. These bijections must be so that for any copy $p'$ of a partial path $p$ of $P$, $\psi_1(p')$ and $\psi_2(p')$ have the same vertices and the same orderings: $\Pi_p=\Pi_{\psi_1(p')}=\Pi_{\psi_2(p')}$. Moreover, for any vertex $x$ in a copy $p'$ of a partial path $p\in P$, if $\varphi_p(x)$ contains ``$\downarrow$", then together $\varphi_{\psi_1(p')}(x)$ and $\varphi_{\psi_2(p')}(x)$ must contain the same number of ``$\downarrow$" as $\varphi_p(x)$. If $\varphi_p(x)$ contains ``$\uparrow$'' or ``$-$'', then each of $\varphi_{\psi_1(p')}(x)$ and $\varphi_{\psi_2(p')}(x)$ must contain the same number of ``$\uparrow$'' and ``$-$'' as $\varphi_p(x)$. Finally, if $\varphi_p(x)=\emptyset$, also $\varphi_{\psi_1(p')}(x)=\varphi_{\psi_2(p')}(x)=\emptyset$.

Let $k$ be the number of partial paths of $P$ that intersect $X_v$, and let $C$ be the collection of all pairs of DP-states $([u_1,P',f'],[u_2,P'',f''])$ compatible with $[v,P,f]$. We have:
    \[
    opt[v,P,f]=\min_{\substack{([u_1,P',f'],[u_2,P'',f''])\in C}}\{opt[u_1,P',f']+opt[u_2,P'',f'']-k\}
    \]

To see the validity of the formula, as before, we argue that given two DP-states $[u_1,P',f']$, $[u_2,P'',f'']$ for $u_1$ and $u_2$ that are compatible with $[v,P,f]$, we can produce a partial solution that corresponds to $(P,f)$ by merging the partial solutions corresponding to $(P',f')$ and $(P'',f'')$. By Lemma~\ref{sepa}, since $X_v$ is a separator of $G$ and the paths of $P,P'$ and $P''$ agree on $X_v$, we indeed obtain a partial solution of $G_v$. This shows the upper bound of the formula.

For the lower bound, conversely, we argue that given any optimal partial solution $S$ corresponding to $(P,f)$, we can find two DP-states $[u_1,P',f']$, $[u_2,P'',f'']$ for $u_1$ and $u_2$ compatible with $[v,P,f]$ such that $S_1$ (the partial solution equal to $S$ restricted to $G_{u_1}$) corresponds to $[u_1,P',f']$, and $S_2$ (the partial solution equal to $S$ restricted to $G_{u_2}$) corresponds to $[u_2,P'',f'']$. Thus, the size of $S$ must be at least $|S_1|+|S_2|-k$. Hence, we have $opt[v,P,f]\geq opt[u_1,P',f']+opt[u_2,P'',f'']-k\geq \min_{\substack{([u_1,P',f'],[u_2,P'',f''])\in C}}\{opt[u_1,P',f']+opt[u_2,P'',f'']-k\}$.

\medskip

We are now ready to prove Theorem~\ref{theoxp}.

\begin{proof}[Proof of Theorem~\ref{theoxp}]
Recall that given the graph $G$ of order $n$ with treewidth $t$, one can compute a nice tree decomposition of width at most $2t+1$ with at most $4n$ nodes in time $2^{O(t)}n$~\cite{bodlaender2013fine,DBLP:conf/focs/Korhonen21}.

Our algorithm goes through this tree decomposition in a bottom-up fashion and computes, for each possible DP-state of the current node, the optimal value for this state using the children's optimal values. The correctness of this algorithm is given by the above discussion and the correctness of the inductive formulas. The optimal value of a solution is then obtained at the root. To obtain the actual path cover, one may use a standard backtracking procedure to build it inductively.

The running time is dominated by the generation, at each node, of all possible DP-states. By the preliminary discussions, we have $|X_v|\leq 2t+2$ for each node $v$ of the tree decomposition, if $G$ is of treewidth $t$. Let us count the number of possible DP-states $[v,P,f]$ for a node $v$. $P$ is a collection of partial paths that covers the vertices of $X_v$.
Each partial path in $P$ is an ordered subset of $X_v$, where each vertex has 10 possible types inside the partial path.
Therefore, there are at most $(10t)!=t^{O(t)}$ possible partial paths. To each of them, we associate a number between $1$ and $\kappa$ using the function $f$. There are at most $\kappa^{(10t)!}$ possible functions $f$, thus, the total number of partial solutions inside each node is at most $\kappa^{t^{O(t)}}$, and we can compute them in this time. We can also check the validity, compatibility, etc of the DP-states in time that is polynomial in the size of a DP-state. In a forget/introduce node, we process all pairs of DP-states coming from the parent and child node, which takes $p^2$ time (if there are at most $p$ possible DP-states per node), and in a join node, we go through all triples of DP-states coming from the parent and the two children nodes, which takes $p^3$ time. In each case, we have $p=\kappa^{t^{O(t)}}$ and thus $p^3$ is still $\kappa^{t^{O(t)}}$. We have at most $4n$ nodes in the tree decomposition; hence, the algorithm solves \pc in time $\kappa^{t^{O(t)}}n$.\end{proof}

\subsection{The case of induced paths and edge-disjoint paths}
The case where the solution paths are required to be induced can be handled easily, indeed, a solution path is induced if and only if its intersection with every bag is induced as well. Thus, it suffices to only consider the DP-states where the partial paths have no unwanted chord inside the bag. Otherwise, the algorithm remains the same.

For the edge-disjoint variants, similarly, it suffices to check for every considered DP-state, whether its partial paths are edge-disjoint. Otherwise, the algorithm remains the same.

\section{\pp on graphs of bounded treewidth}\label{sec:PP}
In \pp, since the solution paths are vertex-disjoint, the number of paths that go through each bag of the tree decomposition is at most the size of the bag.
Hence, a DP-state from the algorithm in Section~\ref{sec:tw} simply needs a collection of partial paths that forms a partition of the bag into ordered sets, and thus, there are at most $t^{O(t)}=2^{O(t\cdot\log t)}$ possible DP-states for every node, since there are at most $(t+1)!=t^{O(t)}$ orderings and $t^{O(t)}$ possible partitions of a bag. Thus, we obtain an FPT algorithm for \pp in time $2^{O(t\cdot\log t)}n$.

We next use the Cut\&Count method described in~\cite{cutandcount} to solve \pp in time $2^{O(t)}n$.
In \cite{cutandcount}, they show that this method uses a randomization technique based on the \emph{Isolation Lemma}~\cite{mulmuley1987matching}, which leads to a \emph{one-sided error Monte-Carlo} algorithm. A one-sided error Monte-Carlo algorithm is a randomized algorithm where a probabilistic bias towards a specific outcome is provided.
To be specific, our algorithm is a false-biased one-sided error Monte-Carlo algorithm. 
Therefore, if the solution to our problem does not exist, then the algorithm always returns "NO" accurately, whereas it does not always return "YES" accurately; it returns "YES" with a reasonably good probability.  

We now define the notions required for the isolation lemma. Let $\mathcal{S}=\{S_1,S_2,\ldots,S_k\}$ be a family of subsets of a universe $U=\{x_1,x_2,\ldots, x_n\}$ and $w:U \rightarrow \mathbb{Z}$ be a weight function. 
\begin{definition}
     The weight function $w$ \emph{isolates} the subset family $\mathcal{S}$ if there exists a unique set $S_i \in \mathcal{S}$ such that $\sum_{x_j \in S_i} w(x_j)$ is minimum.
\end{definition}
The following \emph{Isolation Lemma} by Mulmuley, Vazirani and Vazirani~\cite{mulmuley1987matching} is the key ingredient of the Cut\&Count technique. This probabilistic lemma was originally introduced to give a parallel algorithm for finding a maximum matching in a general graph. The lemma works for any arbitrary set system and finds applications in parallel computation and randomized reductions.
\begin{lemma}[Isolation Lemma {\cite[Lemma 1]{mulmuley1987matching}}]
 If the weights $w(x_i)$ of the elements $x_i \in U$ are chosen uniformly and independently at random from $\{1,2,\ldots,N\}$ then 
     $$Pr[w\ isolates\ \mathcal{S}] \geq 1-\frac{n}{N}.$$     
\end{lemma}
We use the Cut\&Count technique following \cite{cutandcount} for our problem. The technique has two parts: the Cut part and the Count part. We use the notion of \emph{markers} (described in \cite[Section 3.2]{cutandcount}).
Let $\mathcal{S}$ be the solution set of tuples of the form $(P,M)$ for \pp for the given graph $G=(V,E)$, such that $P \subseteq E$ is a path partition of $G$, and $M \subseteq V$ is a set of marked vertices (markers) with $|M| = k$.

In this method, we consider the set of paths in the solution as a subset of the edges, and in the next part we explain in detail how to check whether a subset of edges corresponds to a path partition.
Note that in our case, we do not need to store the sequence of the vertices in a path, as we are looking for the subset of edges which reflects the sequence of the vertices of the path constructed by that particular set of edges. 
 Additionally, in the above setting, it is required that each path in $P$ contains at least one marked vertex from $M$. This ensures that the number of paths in $P$ is at most $k$.

\paragraph{The Cut part} We assign weights to the edges and vertices of $G$ (defining the universe $U = E \cup V$) uniformly and independently at random from $\{1,2,\ldots,N\}$, such that $N=2(|E|+|V|)$. Therefore, the weight function $w$ is defined as $w:E \cup V \rightarrow \{1,2,\ldots,2(|E|+|V|)\}$. 

For a particular integer weight $\omega$, where $0 \leq \omega \leq 2(|E|+|V|)^2$, we define $\mathcal{R}_\omega$, $\mathcal{S}_\omega$, and $\mathcal{C}_\omega$ as follows:
\begin{enumerate}
    \item \label{Rw} Let $\mathcal{R}_\omega$ be the set of tuples of the form $(P,M)$, where $P \subseteq E$ is a path partition of $G$, and $M \subseteq V$ is a set of marked vertices with $|M| = k$. $\mathcal{R}_\omega$ denotes the family of possible candidate solutions corresponding to the weight $\omega$, i.e. $w(P \cup M) = \omega$. To check whether $P$ is a path partition of $G$, we follow the argument below. Let $G[P]$ be the subgraph of $G$ with vertex set $V$ and edge set $P$. Let $c_P'$ and $c_P''$ be the numbers of singleton and non-singleton components in $G[P]$, respectively. Since $P$ is a collection of vertex-disjoint paths and singleton vertices, $deg_{G[P]}(v) \leq 2$ for any vertex $v \in V$. The number of degree $0$ vertices in $G[P]$ is $c_P'$, the number of degree $1$ vertices in $G[P]$ is $2 \cdot c_P''$, and the remaining vertices, i.e. $(n-2 \cdot c_P'' - c_P')$ vertices, are degree $2$ vertices (assuming $n$ is the number of vertices in $G[P]$). To get the value of $c_P'$ and $c_P''$ we can run a BFS or DFS on $G[P]$.  
    \item \label{Sw} Let $\mathcal{S}_\omega$ be the solution set of tuples $(P,M) \in \mathcal{R}_\omega$ such that each component in $G[P]$ contains at least one marked vertex from $M$ (if the component is a singleton vertex then that vertex must be a marked vertex from $M$).
    \item \label{Cw} Let $\mathcal{C}_\omega$ be the set of tuples of the form $((P,M),(V(P_1),V(P_2)))$, where $(P,M) \in \mathcal{R}_\omega$ and $(V(P_1),V(P_2))$ is a \emph{consistent cut} of $G[P]$ and $M \subseteq V(P_1)$ i.e. all the marked vertices $M$ are on one side of the cut. A cut is called a consistent cut if for any two disjoint vertices $u \in V(P_1)$ and $v \in V(P_2)$, the edge $uv \notin E(G)$. 
\end{enumerate}

\begin{proposition}
    $G$ has a path partition of size at most $k$ if and only if there exists a weight value $\omega$ for which $\mathcal{S}_\omega$ is non-empty. 
\end{proposition}
\begin{proof}
According to the definition of $\mathcal{S}_\omega$, if $\mathcal{S}_\omega$ is non-empty that means there exists at least one tuple $(P,M)$ such that $P$ is a path partition, and each path in $P$ contains at least one marked vertex from $M$. 
    Since $|M|=k$, therefore, the number of paths in \pp is at most $k$.

Suppose $G$ has a \pp $P$ of size at most $k$, then from each path in $P$, we mark one vertex and add it to the set $M$.
    Since the paths in $P$ are vertex-disjoint, the size of $M$ is at most $k$.
    Let $\omega$ be the sum of the weights on the edges of the path $P$ and weights on the marked vertices.
    Thus, we can define a non-empty $\mathcal{S}_\omega$ by $P$ and $M$ and corresponding $\omega$.
\end{proof} 

\paragraph{The Count Part}
We now describe the Count part of the Cut\&Count method.
\begin{lemma}
\label{lem:cut}
    Given $G, \omega, \mathcal{R}_\omega, \mathcal{S}_\omega,$ and $\mathcal{C}_\omega$, for each weight value $\omega$, we have $|\mathcal{S}_\omega| \equiv |\mathcal{C}_\omega| \pmod{2}$.
\end{lemma}
\begin{proof}
For a tuple $(P,M) \in \mathcal{R}_\omega$, let $c_P$ be the number of connected components of $G[P]$ that do not contain any vertex $v \in M$. 
\\Given a tuple $((P,M),(V(P_1),V(P_2))) \in \mathcal{C}_\omega$, all the marked vertices $v \in M$ are in $V(P_1)$. Therefore, the components containing no vertices of $M$ are either part of $V(P_1)$ or $V(P_2)$.
Hence, for a fixed tuple $(P,M) \in \mathcal{R}_\omega$, there exist $2^{c_P}$ tuples in $\mathcal{C}_\omega$, by varying the consistent cuts $V(P_1)$ and $V(P_2)$. Therefore, $|\mathcal{C}_\omega| = \sum_{\forall (P,M) \in \mathcal{R}_\omega} 2^{c_P}$. Since $\mathcal{S}_\omega$ is the solution set of tuples $(P,M) \in \mathcal{R}_\omega$ such that each connected component of $G[P]$ contains at least one marked vertex from $M$, a specific tuple $(P,M) \in \mathcal{S}_\omega$ corresponds to those tuples of $\mathcal{C}_\omega$ where $c_P = 0$, or equivalently $V(P_2) = \emptyset$. Therefore, $|\mathcal{S}_\omega| \equiv |\mathcal{C}_\omega| \pmod{2}$.
\end{proof}

The type of tree decomposition we use is the same as described in~\cite{cutandcount}. We call this type of tree decomposition an \emph{advanced nice tree decomposition} where along with the properties of a nice tree decomposition, an additional type of node is used. These nodes are called \emph{introduce edge nodes}.
\begin{definition}
    An \emph{introduce edge node} $x$ has exactly one child $y$ such that $X_x=X_y$. The node is labeled with an edge $uv\in E$ with $u,v\in X_x=X_y$.
\end{definition}
These nodes are 
present for every edge $uv \in E(G)$. Let $\mathcal{T}$ be an advanced nice tree decomposition, and $x$ be an introduce edge node for an edge $uv \in E(G)$ with a child bag $y$. Then $u, v \in X_x$ and $X_x = X_y$, indicating that the edge $uv$ is introduced in the bag of $x$. In a nice tree decomposition, when a vertex is introduced in an introduce node, all the edges that are incident to that vertex are also considered in the computation when we process the introduce node. However, for an advanced nice tree decomposition, there are separate bags (introduce vertex nodes and introduce edge nodes) that introduce a vertex and the edges incident to it separately. Therefore, when we process an introduce vertex node, the vertex that has been introduced in the corresponding bag remains isolated, and in later stages, with the help of introduce edge nodes, all the edges incident to that vertex are processed eventually, one at a time. The other node types in an advanced nice tree decomposition are the same as in Definition~\ref{nicetree}. It is noted in~\cite{cutandcount} that, given a tree decomposition, an advanced nice tree decomposition of the same width can be computed in polynomial time.    

\begin{lemma}
\label{lem:count}
    Given a graph $G=(V,E)$, an advanced nice tree decomposition $\mathcal{T}$ of $G$, a weight function $w \rightarrow \{1,2, \ldots,N\}$, and an integer $k$, there exists an algorithm that computes $|\mathcal{C}_\omega| \pmod{2}$ for all $0 \leq \omega \leq (k+ |V|)N$ in $|V|^{O(1)} N^2 25^t$ time.
\end{lemma}

\begin{proof}
Let $X_x$ be the set of vertices in bag $x$.
Let $G_x$ be a subgraph of $G$ induced by vertex set $V_x$, where $V_x$ is the union of all the bags in the subtree of $\mathcal{T}$ rooted at a bag $x$ (including the vertices present in $x$), and the edge set $E_x$ be the set of edges of $G$ with both endpoints in $V_x$. The partial solutions of the problem are defined below. For each bag $x \in \mathcal{T}$, every integer $1 \leq i \leq k$, and for all weight values $0 \leq \omega \leq (k+ |V|)N$, we define $\mathcal{A}_x(i,\omega,s)$ as a collection of tuples $((P,M),(V(P_1),V(P_2)))$ 
with the following properties: 
    \begin{enumerate}
        \item $P \subseteq E_x$, and $M \subseteq V(P_1) \backslash X_x$. 
        \item $(V(P_1),V(P_2))$ be the consistent cut of $G[P]$, where $G[P]$ is a subgraph of $G$ with vertex set $V_x$ and edge set $P$.
        \item For every $v \in V_x$, $deg_{G[P]}(v) \leq 2$.
        \item $P$ is a path partition of $G[P]$ (follow the definition of $\mathcal{R}_\omega$ for details).
        \item $w(X) + w(M) = \omega$, and $|M| = i$.
        \item $s: X_x \rightarrow \{0_1, 0_2, 1_1, 1_2, 2\}$ is a function defined for each vertex in $x$ and that assumes values based on their degree in $G[P]$. The value $s(v)$ of a vertex $v \in x$ is of the form $\alpha_\beta$, where $\alpha$ denotes the degree of $v$ in $G[P]$ with values either $0$ (singleton set, not included in a path) or $1$ (endpoint of a path) or $2$ (internal vertex of a path) and $\beta$ denotes which side of the consistent cut $v$ is present with values either $1$ (when $v \in V(P_1)$) or $2$ (when $v \in V(P_2)$). 
        \end{enumerate}
Note that when a vertex is forgotten in the forget node, then it can be made a marked vertex by properly adjusting the value of $i$. Moreover, if $deg_{G[P]}(v) = 2$ then $v$ has been processed and is no longer required in the future since $v$ is an internal vertex of a path. Therefore, it is not necessary to store which side of the cut $v$ is in such a case. Since we consider singleton paths in our calculation, it is important to store which side of the cut a vertex $v$ is in when $deg_{G[P]}(v) = 0$.

Let $r$ be the root node of $\mathcal{T}$. Since $|\mathcal{C}_\omega| \pmod 2 = \mathcal{A}_r(i,\omega,\emptyset) \pmod 2$, to compute the value of $|\mathcal{C}_\omega| \pmod 2$ for all $\omega$ it is enough to compute $\mathcal{A}_r(i,\omega,\emptyset)$ for all $\omega$.
\paragraph{\textbf{Introduce vertex node}} Let $v$ be the vertex introduced in the introduce vertex bag. Since none of the edges incident to $v$ are processed at this stage, $v$ remains a singleton path.
$$\mathcal{A}_x(i, \omega, s[v \rightarrow 0_j]) = \mathcal{A}_y(i,\omega, s) \qquad \text{for }j \in \{0,1\}.
$$
\paragraph{\textbf{Introduce edge node}} We refer to  \cite[Theorem A.14]{cutandcount} for this part. The only thing that is different for our problem is the states for the degree $0$ vertices (either $0_1$ or $0_2$). Following the definition of $subs(\alpha)$ in \cite{cutandcount} we can write  $subs(0_1) = subs(0_2) = \emptyset$. The rest of the analysis is the same as for the introduce edge node in \cite[Theorem A.14]{cutandcount}.  

\paragraph{\textbf{Forget Node}}
 Let $v$ be the vertex that has been forgotten in the forget node. Then $v$ can be a singleton vertex, or an internal vertex of a path, or an endpoint of a path. Hence, $$\mathcal{A}_x(i,w,s) = \sum \limits_ {\alpha \in \substack{\{0_1,0_2,\\ 1_1,1_2,2\}}}\mathcal{A}_y(i,w,s[v \rightarrow \alpha]).$$ 

 \paragraph{\textbf{Join node}} We again refer to \cite[Theorem A.14]{cutandcount} for this part. Similar to the introduce edge node, the two states of a degree $0$  vertex are handled separately.

 A standard dynamic programming to compute all the values of $\mathcal{A}_x(i,\omega,s)$ takes time $|V|^{O(1)} N^2 25^t$, where the number $25 = 5^2$ comes from the calculation of the join node. (Note that one could use fast subset convolution like in \cite{cutandcount} to reduce the constant value $25$, which we do not do here.)
\end{proof}

\begin{theorem}
There exists a Monte-Carlo algorithm that, given a graph $G$ and a tree decomposition of width $t$ of $G$, solves \pp in $25^t |V|^ {O(1)}$ time. The algorithm does not give any false positives and gives false negatives with
probability at most $1/2$.
\end{theorem}
\begin{proof}
    Combining Algorithm 1 (by setting $U = E \cup V$), Corollary 3.1 of \cite{cutandcount}, and Lemma~\ref{lem:count}, we get the algorithm to solve \pp. The correctness of the algorithm and the running time follow from Lemma~\ref{lem:cut} and Lemma~\ref{lem:count}, respectively. 
\end{proof}

\section{Conclusion}\label{sec:conclu}

We have re-initiated the study of \pc, which surprisingly is not extensively studied. We settled its complexity for trees by giving a linear-time algorithm on this class, and we gave an explicit $n^{t^{O(t)}}$ XP-time dynamic programming scheme for graphs of treewidth $t$, or $\kappa^{t^{O(t)}}n$ if we have an upper bound on the solution size $\kappa$. An algorithm for \pp using the Cut\&Count technique gives a $2^{O(t)}n$ FPT running time. These running times also hold for the variants of \pc and \pp where the solution paths are required to be induced and/or edge-disjoint.

It would be nice to improve the running times of our algorithm for \pc. Can one get an improved $n^{O(t)}$ algorithm? Can it actually be solved in FPT time for parameter treewidth alone? It is possible that this is not the case, as the number of solution paths going through one bag of the tree decomposition can be arbitrarily large (see Figure~\ref{bubble}). In case of a negative answer, this would show a striking contrast between the algorithmic complexities of \pc and \pp.

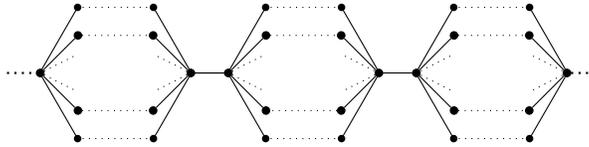
\begin{figure}[H]\label{bubble}
    \centering
    \begin{tikzpicture}
        \foreach \i in {0,1,2}
        {
            \node[regular polygon, regular polygon sides=6, minimum size=2cm] (hex\i) at (\i*2.5, 0) {}; 
        }

        \foreach \i in {0,1,2}
        {
            \foreach \j in {1,...,6}
            {
                \node[circle, fill=black, inner sep=1pt] at (hex\i.corner \j) {}; 
            }
        }

        \draw [dotted] (hex0.corner 1) -- (hex0.corner 2);
        \draw [dotted] (hex1.corner 1) -- (hex1.corner 2);
        \draw [dotted] (hex2.corner 1) -- (hex2.corner 2);
        \draw (hex0.corner 1) -- (hex0.corner 6);
        \draw (hex1.corner 1) -- (hex1.corner 6);
        \draw (hex2.corner 1) -- (hex2.corner 6);
        \draw (hex0.corner 3) -- (hex0.corner 2);
        \draw (hex1.corner 3) -- (hex1.corner 2);
        \draw (hex2.corner 3) -- (hex2.corner 2);
        \draw (hex0.corner 3) -- (hex0.corner 4);
        \draw (hex1.corner 3) -- (hex1.corner 4);
        \draw (hex2.corner 3) -- (hex2.corner 4);        
        \draw (hex0.corner 5) -- (hex0.corner 6);
        \draw (hex1.corner 5) -- (hex1.corner 6);
        \draw (hex2.corner 5) -- (hex2.corner 6);
        \draw[dotted] (hex0.corner 5) -- (hex0.corner 4);
        \draw[dotted] (hex1.corner 5) -- (hex1.corner 4);
        \draw[dotted] (hex2.corner 5) -- (hex2.corner 4);

        \node[circle, draw=white, fill=white, inner sep=1pt] (z) at (-1.5,0) {};
        \node[circle, draw, fill=black, inner sep=1pt] (a) at (-1,0) {};
        \node[circle, draw=white, fill=white, inner sep=1pt] (b) at (-0.5,.25) {};
        \node[circle, draw=white, fill=white, inner sep=1pt] (b') at (-0.5,-0.25) {};
        \node[circle, draw=white, fill=white, inner sep=1pt] (b'') at (0.5,0.25) {};
        \node[circle, draw=white, fill=white, inner sep=1pt] (b''') at (0.5,-0.25) {};
        \node[circle, draw, fill=black, inner sep=1pt] (c) at (1,0) {};
        \node[circle, draw, fill=black, inner sep=1pt] (c') at (1.5,0) {};
        \node[circle, draw=white, fill=white, inner sep=1pt] (d) at (2,.25) {};
        \node[circle, draw=white, fill=white, inner sep=1pt] (d') at (2,-.25) {};
        \node[circle, draw=white, fill=white, inner sep=1pt] (e) at (3,.25) {};
        \node[circle, draw=white, fill=white, inner sep=1pt] (e'') at (3,-0.25) {};
        \node[circle, draw, fill=black, inner sep=1pt] (e') at (3.5,0) {};
        \node[circle, draw, fill=black, inner sep=1pt] (f) at (4,0) {};
        \node[circle, draw=white, fill=white, inner sep=1pt] (g) at (4.5,.25) {};
        \node[circle, draw=white, fill=white, inner sep=1pt] (g') at (4.5,-.25) {};
        \node[circle, draw=white, fill=white, inner sep=1pt] (h) at (5.5,.25) {};
        \node[circle, draw=white, fill=white, inner sep=1pt] (h') at (5.5,-.25) {};
        \node[circle, draw=black, fill=black, inner sep=1pt] (p) at (6,0) {};
        \node[circle, draw=white, fill=white, inner sep=1pt] (i) at (6.5,0) {};
        \draw[thick,dotted](p) -- (i);
        \draw[thick,dotted](z) -- (a);
        
        \draw[dotted](a)--(b');
        \draw[dotted](a)--(b);
        \draw[dotted](c)--(b'');
        \draw[dotted](c)--(b''');
        \draw[dotted](c')--(d');
        \draw[dotted](c')--(d);
        \draw[dotted](e')--(e'');
        \draw[dotted](e)--(e');
        \draw[dotted](g)--(f);
        \draw[dotted](g')--(f);
        \draw[dotted](h)--(p);
        \draw[dotted](h')--(p);
        
        \node[circle, draw, fill=black, inner sep=1pt] (j) at (-0.5,.5) {};
        \node[circle, draw, fill=black, inner sep=1pt] (k) at (0.5,.5) {};
        \node[circle, draw, fill=black, inner sep=1pt] (l) at (2,.5) {};
        \node[circle, draw, fill=black, inner sep=1pt] (m) at (3,.5) {};
        \node[circle, draw, fill=black, inner sep=1pt] (n) at (4.5,.5) {};
        \node[circle, draw, fill=black, inner sep=1pt] (o) at (5.5,.5) {};
        \draw (a) -- (j);
        \draw (k)--(c)--(c')--(l);
        \draw(m)--(e')--(f)--(n);
        \draw(o)--(p);
        \draw[dotted] (j)--(k);
        \draw[dotted] (l)--(m);
        \draw[dotted] (n)--(o);
        \node[circle, draw, fill=black, inner sep=1pt] (v) at (-0.5,-.5) {};
        \node[circle, draw, fill=black, inner sep=1pt] (q) at (0.5,-.5) {};
        \node[circle, draw, fill=black, inner sep=1pt] (r) at (2,-.5) {};
        \node[circle, draw, fill=black, inner sep=1pt] (s) at (3,-.5) {};
        \node[circle, draw, fill=black, inner sep=1pt] (t) at (4.5,-.5) {};
        \node[circle, draw, fill=black, inner sep=1pt] (u) at (5.5,-.5) {};
        \draw (a) -- (v);
        \draw (q)--(c)--(c')--(r);
        \draw (s)--(e');
        \draw (f)--(t);
        \draw(u)--(p);
        \draw[dotted] (s)--(r);
        \draw[dotted] (v)--(q);
        \draw[dotted] (t)--(u);
        
    \end{tikzpicture}
    \caption{A graph of arbitrarily large order $n$ and treewidth~2, with $\Theta(\sqrt{n})$ cut-vertices, and 2-connected components each of order $\Theta(\sqrt{n})$. An optimal solution to \pc consists of $\Theta(\sqrt{n})$ solution paths, each of length $\Theta(\sqrt{n})$, all going from the left to the right side of the graph. Thus, there is an unbounded number of solution paths going through every bag of any tree decomposition of bounded width.}
    \label{fig:example}
\end{figure}

\bibliographystyle{abbrv}
\bibliography{references}

\end{document}